\documentclass[11pt]{article}
\usepackage{fullpage}
\usepackage{amsthm,amssymb,amsmath}  
\usepackage{xspace,enumerate}
\usepackage[utf8]{inputenc}
\usepackage{thmtools}
\usepackage{thm-restate}
\usepackage{authblk}

  \theoremstyle{plain}
  \newtheorem{theorem}{Theorem}
  \newtheorem{lemma}[theorem]{Lemma}  
   
  \newtheorem{proposition}[theorem]{Proposition}  
  
  \newtheorem{observation}[theorem]{Observation}
  \theoremstyle{definition}

  \newtheorem{example}[theorem]{Example}

  \newtheorem*{claim}{Claim}

  \usepackage{microtype}
  \usepackage{tikz}
  \usepackage{xspace}
  \usepackage{makecell}
  \usepackage{amsmath}
  \usepackage{environ}
  \usepackage[ruled,noline,noend]{algorithm2e}
  \SetKw{Continue}{continue}
  \usepackage{tabularx}
  \usepackage{comment}
  \usepackage{todonotes}
  \usepackage{enumerate}
  \usepackage{mathtools}
  \usepackage{setspace}
  \usepackage[hidelinks]{hyperref}
  \usepackage[capitalise]{cleveref}
  \SetCommentSty{textnormal}

\newcommand{\defproblem}[3]{
  \vspace{3mm}
  \noindent\fbox{
  \begin{minipage}{0.96\textwidth}
  \textsc{#1}

  \smallskip
  \noindent
  {\bf{Input:}} #2
  
  \smallskip
  \noindent
  {\bf{Output:}} #3
  \end{minipage}
  }
  \vspace{3mm}
}

  \newcommand{\floor}[1]{\left\lfloor #1 \right\rfloor}
  
  \newcommand{\Oh}{\mathcal{O}}
  \newcommand{\Ohstar}{\Oh^*}
  \newcommand{\dd}{\mathinner{.\,.}}

  \newcommand{\WSCSFull}{\textsc{Weighted Shortest Common Supersequence}\xspace}
  \newcommand{\WLCSFull}{\textsc{Weighted Longest Common Subsequence}\xspace}
  \newcommand{\SCSFull}{\textsc{Shortest Common Supersequence}\xspace}
  \newcommand{\LCSFull}{\textsc{Longest Common Subsequence}\xspace}
  \newcommand{\WSCS}{\textsc{WSCS}\xspace}
  \newcommand{\Knapsack}{\textsc{Knapsack}\xspace}
  \newcommand{\HS}{\textsc{Merge}\xspace}
  \newcommand{\WLCS}{\textsc{WLCS}\xspace}
  \newcommand{\SCS}{\textsc{SCS}\xspace}
  \newcommand{\LCS}{\textsc{LCS}\xspace}
  \newcommand{\WC}{\textsc{Weighted Consensus}\xspace}
  \newcommand{\SubsetSum}{\textsc{Subset Sum}\xspace}
  
  \newcommand{\DP}{\mathbf{DP}}
  \newcommand{\DPR}{\overrightarrow{\mathbf{DP}}}
  \newcommand{\DPL}{\overleftarrow{\mathbf{DP}}}
  \newcommand{\LR}{\overrightarrow{L}}
  \newcommand{\LL}{\overleftarrow{L}}
  \newcommand{\TT}{\overleftarrow{T}}

  \newcommand{\Freq}{\mathit{Freq}}
  \newcommand{\RFreq}{\overrightarrow{\mathit{Freq}}}
  \newcommand{\LFreq}{\overleftarrow{\mathit{Freq}}}

  \renewcommand{\P}{\mathcal{P}}
  \renewcommand{\H}{\mathcal{H}}

  \newcommand{\fr}{\ensuremath{\tfrac1z}\xspace}
  \newcommand{\match}[1]{\approx_{#1}}
  \newcommand{\mm}[1]{\approx_{#1}}
  \newcommand{\subseq}[1]{\subseteq_{#1}}
  \newcommand{\R}[1]{\subseteq_{#1}}
  \newcommand{\Match}{\mathsf{Matched}}

\title{Weighted Shortest Common Supersequence Problem Revisited}

\newcommand*\samethanks[1][\value{footnote}]{\footnotemark[#1]}

\author[1]{Panagiotis Charalampopoulos}
\author[2,3]{Tomasz Kociumaka\thanks{Supported by ISF grants no. 824/17 and 1278/16 and by an ERC grant MPM under the EU's Horizon 2020 Research and Innovation Programme (grant no. 683064).}}
\author[4]{Solon P. Pissis}
\author[3]{Jakub Radoszewski\thanks{Supported by the ``Algorithms for text processing with errors and uncertainties'' project carried out within the HOMING program of the Foundation for Polish Science co-financed by the European Union under the European Regional Development Fund.}}
\author[3]{Wojciech Rytter}
\author[3]{Juliusz Straszyński\samethanks[2]}
\author[3]{Tomasz Waleń}
\author[3]{Wiktor Zuba}

\affil[1]{Department of Informatics, King's College London, London, UK\\
    \texttt{panagiotis.charalampopoulos@kcl.ac.uk}}
\affil[2]{Department of Computer Science, Bar-Ilan University, Ramat Gan, Israel}
\affil[3]{Institute of Informatics, University of Warsaw, Warsaw, Poland\\
    \texttt{$\{$kociumaka,jrad,rytter,jks,walen,w.zuba$\}$@mimuw.edu.pl}}
\affil[4]{CWI, Amsterdam, The Netherlands\\
    \texttt{solon.pissis@cwi.nl}}
\date{\vspace{-5ex}}

\begin{document}

\maketitle

\begin{abstract}
A weighted string, also known as a position weight matrix, is a sequence of probability distributions over some alphabet.
We revisit the Weighted Shortest Common Supersequence (WSCS) problem, introduced by Amir et al.\ [SPIRE 2011], that is, the SCS problem on weighted strings. 
In the WSCS problem, we are given two weighted strings $W_1$ and $W_2$ and a threshold $\fr$ on probability,
and we are asked to compute the shortest (standard) string $S$ such that both $W_1$ and $W_2$ match subsequences of $S$
(not necessarily the same) with probability at least $\fr$.
Amir et al.\ showed that this problem is NP-complete if the probabilities, including the threshold $\fr$,
are represented by their logarithms (encoded in binary).

We present an algorithm that solves the WSCS problem for two weighted strings of length $n$
over a constant-sized alphabet in $\Oh(n^2\sqrt{z} \log{z})$ time.
Notably, our upper bound matches known conditional lower bounds stating that
the WSCS problem cannot be solved in $\Oh(n^{2-\varepsilon})$ time or in $\Ohstar(z^{0.5-\varepsilon})$ time%
\footnote{The $\Ohstar$ notation suppresses factors polynomial with respect to the instance size (with numeric values encoded in binary).}
unless there is a breakthrough improving upon long-standing upper bounds for fundamental NP-hard problems (\textsc{CNF-SAT} and \SubsetSum, respectively).

We also discover a fundamental difference between the WSCS problem and the
Weighted Longest Common Subsequence (WLCS) problem, introduced by Amir et al.\ [JDA 2010].
We show that the WLCS problem cannot be solved in $\Oh(n^{f(z)})$ time, for any function $f(z)$, unless $\mathrm{P}=\mathrm{NP}$.
\end{abstract}
\newpage

\section{Introduction}

Consider two strings $X$ and $Y$. 
A common supersequence of $X$ and $Y$ is a string $S$ such that $X$ and $Y$ are both subsequences of $S$. 
A shortest common supersequence (\SCS) of $X$ and $Y$ is a common supersequence of $X$ and $Y$ of minimum length.
The \SCSFull problem (the \SCS problem, in short) is to compute an \SCS of $X$ and $Y$.
The \SCS problem is a classic problem in theoretical computer science~\cite{DBLP:journals/jacm/Maier78,DBLP:journals/tcs/RaihaU81,DBLP:journals/siamcomp/JiangL95}. 
It is solvable in quadratic time using a standard dynamic-programming approach~\cite{DBLP:books/daglib/0023376},
which also allows computing a shortest common supersequence of any constant number of strings (rather than just two) in polynomial time.
In case of an arbitrary number of input strings, the problem becomes NP-hard~\cite{DBLP:journals/jacm/Maier78} 
even when the strings are binary~\cite{DBLP:journals/tcs/RaihaU81}.

A weighted string of length $n$ over some alphabet $\Sigma$ is a type of uncertain sequence. 
The uncertainty at any position of the sequence is modeled using a subset of the alphabet (instead of a single letter),
with every element of this subset being associated with an occurrence probability;
the probabilities are often represented in an $n \times |\Sigma|$ matrix.
These kinds of data are common in various applications where: (i) imprecise data measurements are recorded; (ii) flexible sequence modeling, such as binding profiles of molecular sequences, 
is required; (iii) observations are private and thus sequences of observations may have artificial uncertainty introduced deliberately~\cite{DBLP:journals/tkde/AggarwalY09}.
For instance, in computational biology they are known as position weight matrices or position probability matrices~\cite{Stormo_1982}. 

In this paper, we study the \WSCSFull problem (the \WSCS problem, in short) introduced by Amir et al.~\cite{DBLP:conf/spire/AmirGS11}, which is a generalization of the \SCS problem for weighted strings. 
In the \WSCS problem, we are given two weighted strings $W_1$ and $W_2$ and a probability threshold $\fr$,
and the task is to compute the shortest (standard) string such that both $W_1$ and $W_2$ match subsequences of $S$
(not necessarily the same) with probability at least $\fr$. 
In this work, we show the first efficient algorithm for the \WSCS problem.

A related problem is the \WLCSFull problem (the \WLCS problem, in short). 
It was introduced by Amir et al.~\cite{DBLP:journals/jda/AmirGS10} and further studied 
in~\cite{DBLP:journals/dam/CyganKRRW16} and, very recently, in~\cite{DBLP:journals/corr/abs-1901-04068}.
In the \WLCS problem, we are also given two weighted strings $W_1$ and $W_2$ and a threshold $\fr$ on probability,
but the task is to compute the longest (standard) string $S$ such that
$S$ matches a subsequence of $W_1$ with probability at least $\fr$
and
$S$ matches a subsequence of $W_2$ with probability at least $\fr$.
For standard strings $S_1$ and $S_2$,
the length of their shortest common supersequence $|\SCS(S_1,S_2)|$
and
the length of their longest common subsequence $|\LCS(S_1,S_2)|$ 
satisfy the following folklore relation:
\begin{equation}\label{eq:LCS_SCS}
  |\LCS(S_1,S_2)| + |\SCS(S_1,S_2)| = |S_1| + |S_2|.
\end{equation}
However, an analogous relation does not connect the \WLCS and \WSCS problems, even though
both problems are NP-complete because of similar reductions, which remain valid even in the case that both weighted strings have the same length~\cite{DBLP:journals/jda/AmirGS10,DBLP:conf/spire/AmirGS11}.
In this work, we discover an important difference between the two problems.

Kociumaka et al.~\cite{DBLP:journals/mst/KociumakaPR19} introduced a problem called \WC,
which is a special case of the \WSCS problem asking whether the \WSCS of two weighted strings of length $n$ is of length $n$,
and they showed that the \WC problem is NP-complete yet
admits an algorithm running in pseudo-polynomial time $\Oh(n+\sqrt{z}\log z)$ for constant-sized alphabets\footnote{Note that in general $z\notin \Ohstar(1)$ unless $z$ is encoded in unary.}.
Furthermore, it was shown in \cite{DBLP:journals/mst/KociumakaPR19} that the \WC problem cannot be solved in $\Ohstar(z^{0.5-\varepsilon})$ time
for any $\varepsilon>0$ unless there is an $\Ohstar(2^{(0.5-\varepsilon)n})$-time algorithm for the \SubsetSum problem.
Let us recall that the \SubsetSum problem, for a set of $n$ integers, asks whether there is a subset summing up to a given integer.
Moreover, the $\Ohstar(2^{n/2})$ running time for the \SubsetSum problem, achieved by a classic meet-in-the-middle
approach of Horowitz and Sahni~\cite{DBLP:journals/jacm/HorowitzS74}, has not been improved yet despite much effort; see e.g.~\cite{DBLP:journals/siamcomp/BansalGN018}.

Abboud et al.~\cite{DBLP:conf/focs/AbboudBW15} showed that the \LCSFull problem over constant-sized alphabets cannot be solved in $\Oh(n^{2-\varepsilon})$
time for $\varepsilon>0$ unless the Strong Exponential Time Hypothesis
\cite{DBLP:journals/jcss/ImpagliazzoP01,DBLP:journals/jcss/ImpagliazzoPZ01,DBLP:journals/eatcs/LokshtanovMS11} fails.
By~\eqref{eq:LCS_SCS}, the same conditional lower bound applies to the \SCS problem, and since standard strings are a special
case of weighted strings (having one letter occurring with probability equal to 1 at each position), it also applies to the \WSCS problem.

The following theorem summarizes the above conditional lower bounds on the \WSCS problem.

\begin{theorem}[Conditional hardness of the \WSCS problem; see \cite{DBLP:conf/focs/AbboudBW15,DBLP:journals/mst/KociumakaPR19}]\label{the:CLB}
Even in the case of constant-sized alphabets, the \WSCSFull problem is NP-complete, and for any $\varepsilon>0$ it cannot be solved:
\begin{enumerate}
\item in $\Oh(n^{2-\varepsilon})$ time unless the Strong Exponential Time Hypothesis fails;
\item in $\Ohstar(z^{0.5-\varepsilon})$ time unless there is an
$\Ohstar(2^{(0.5-\varepsilon)n})$-time algorithm for the \SubsetSum problem.
\end{enumerate}
\end{theorem}

\paragraph{Our Results}
We give an algorithm for the \WSCS problem with pseudo-polynomial running time that depends polynomially on $n$ and $z$. 
Note that such algorithms have already been proposed for several problems on weighted strings: pattern matching~\cite{DBLP:journals/tcs/BartonLP16,DBLP:journals/iandc/Charalampopoulos19,DBLP:journals/mst/KociumakaPR19,DBLP:conf/dcc/RadoszewskiS17}, 
indexing~\cite{DBLP:journals/tcs/AmirCIKZ08,DBLP:conf/cpm/BartonKPR16,zstrings,DBLP:conf/latin/Charalampopoulos18}, 
and finding regularities~\cite{DBLP:journals/algorithmica/BartonP18}.
In contrast, we show that no such algorithm is likely to exist for the \WLCS problem. 

Specifically, we develop an $\Oh(n^2 \sqrt{z} \log{z})$-time algorithm for the \WSCS problem in the case of a constant-sized alphabet\footnote{We consider the case of $|\Sigma|=\Oh(1)$ just for simplicity. For a general alphabet, our algorithm can be modified to work
in $\Oh(n^2 |\Sigma|\sqrt{z} \log{z})$ time.}.
This upper bound matches the conditional lower bounds of Theorem~\ref{the:CLB}.
We then show that unless $P=NP$, the \WLCS problem cannot be solved in $\Oh(n^{f(z)})$ time for any function $f(\cdot)$.

\paragraph{Model of Computations}
We assume the word RAM model with word size $w = \Omega(\log n + \log z)$.
We consider the log-probability representation of weighted sequences, that is, we assume that
the non-zero probabilities in the weighted sequences and the threshold probability \fr are all of the form $c^{\frac{p}{2^{dw}}}$,
where $c$ and $d$ are constants and $p$ is an integer that fits in $\Oh(1)$ machine words.

\section{Preliminaries}
A \emph{weighted string} $W=W[1] \cdots W[n]$ of length $|W|=n$ over alphabet $\Sigma$
is a sequence of sets of the form
\[W[i] = \{(c,\ \pi^{(W)}_i(c))\ :\ c \in \Sigma\}.\]
Here, $\pi_i^{(W)}(c)$ is the occurrence probability of the letter $c$ at the position $i \in [1\dd n]$.%
\footnote{For any two integers $\ell\le r$, we use $[\ell\dd r]$ to denote the integer range $\{\ell,\ldots,r\}$. }
These values are non-negative and sum up to 1 for a given index $i$.

By $W[i \dd j]$ we denote the weighted \emph{substring}  $W[i]\cdots W[j]$;
it is called a prefix if $i=1$ and a suffix if $j=|W|$.

The \emph{probability of matching} of a string $S$ with a weighted string $W$, with $|S|=|W|=n$, is
\[\P(S,W) \, =\, \prod_{i=1}^n \pi^{(W)}_i(S[i])\,=\, \prod_{i=1}^n\, \P(S[i]=W[i]).\]
We say that a (standard) string $S$ \emph{matches a weighted string $W$
with probability at least $\fr$}, denoted by $S \mm{z} W$, if $\P(S,W) \ge \fr$.
We also denote \[\Match_z(W)=\{S \in \Sigma^n : \P(S,W)\ge \fr\}.\]

\smallskip
For a string $S$ we write $W\R{z}S$ if $S'\mm{z}W$ for some 
 subsequence $S'$ of $S$. Similarly we write $S\R{z}W$ if $S\mm{z}W'$ for some
 subsequence $W'$ of $W$.


Our main problem can be stated as follows.

\defproblem{\WSCSFull ($\WSCS(W_1,W_2,z)$)}{
  Weighted strings $W_1$ and $W_2$ of length up to $n$ and a threshold $\fr$.
}{
  A shortest standard string $S$ such that 
$W_1\subseq{z}S$ and $W_2\subseq{z}S$.
}

\begin{example}\label{ex:1}
If the alphabet is $\Sigma=\{\mathtt{a},\mathtt{b}\}$, then we write the weighted string as $W=[p_1,p_2,\ldots,p_n]$, where
$p_i=\pi^{(W)}_i(\mathtt{a})$; in other words, $p_i$ is the probability that the $i$th letter $W[i]$ is $\mathtt{a}$.
For \[W_1=[1,\, 0.2,\, 0.5], \; W_2=[0.2,\, 0.5,\, 1],\text{ and }z=\tfrac{5}{2},\] we have
$\WSCS(W_1,\, W_2,\,z)=\mathtt{baba}$ since 
$W_1\subseq{z} \mathtt{b\underline{a}\underline{b}\underline{a}},\, 
W_2\subseq{z} \mathtt{\underline{b}\underline{a}b\underline{a}}$ (the witness subsequences are underlined), and
$\mathtt{baba}$ is a shortest string with this property.
\end{example}

We first show a simple solution to \WSCS based on the following facts.

\begin{observation}[Amir et al.~\cite{DBLP:journals/tcs/AmirCIKZ08}]\label{obs:obv}
  Every weighted string $W$ matches at most $z$ standard strings with probability at least \fr,
  i.e.,   $|\Match_z(W)| \le z$.
\end{observation}

\begin{lemma}\label{lem:mzw}
  The set $\Match_z(W)$ can be computed in $\Oh(nz)$ time if $|\Sigma|=\Oh(1)$.
\end{lemma}
\begin{proof}
  If $S \in \Match_z(W)$, then $S[1\dd i] \in \Match_z(W[1 \dd i])$ for every index~$i$.
  Hence, the algorithm computes the sets $\Match_z$ for subsequent prefixes of $W$.
  Each string $S\in \Match_z(W[1 \dd i])$ is represented as a triple $(c,p,S')$, where $c=S[i]$ is the last letter
  of $S$, $p = \P(S,W[1 \dd i])$, and $S'=S[1\dd i-1]$ points to an element of $\Match_z(W[1 \dd i-1])$.
  Such a triple is represented in $\Oh(1)$ space.
  
  Assume that $\Match_z(W[1 \dd i-1])$ has already been computed.
  Then, for every $S'=(c',p',S'') \in \Match_z(W[1 \dd i-1])$ and every $c \in \Sigma$, if $p :=p'\cdot  \pi^{(W)}_i(c) \ge\fr$,
  then the algorithm adds $(c,p,S')$ to $\Match_z(W[1 \dd i])$.

  By Observation~\ref{obs:obv}, $|\Match_z(W[1 \dd i-1])| \le z$ and $|\Match_z(W[1 \dd i])| \le z$.
  Hence, the $\Oh(nz)$ time complexity follows.
\end{proof}

\begin{proposition}\label{prop:1}
  The \WSCS problem can be solved in $\Oh(n^2z^2)$ time if $|\Sigma|=\Oh(1)$.
\end{proposition}
\begin{proof}
  The algorithm builds $\Match_z(W_1)$ and $\Match_z(W_2)$ using Lemma~\ref{lem:mzw}.
  These sets have size at most $z$ by Observation~\ref{obs:obv}.
  The result is the shortest string in
  \[\{\SCS(S_1,S_2)\,:\,S_1 \in \Match_z(W_1),\,S_2\in\Match_z(W_2)\}.\]
  Recall that the \SCS
  of two strings can be computed in $\Oh(n^2)$ time using a standard dynamic programming algorithm~\cite{DBLP:books/daglib/0023376}.
\end{proof}

\noindent We substantially improve upon this upper bound in \cref{sec:DP,sec:improve}.

\subsection{Meet-in-the-middle Technique}
In the decision version of the \Knapsack problem, we are given $n$ items with weights $w_i$ and values $v_i$,
and we seek for a subset of items with total weight up to $W$ and total value at least $V$.
In the classic meet-in-the-middle solution to the \Knapsack problem by
Horowitz and Sahni~\cite{DBLP:journals/jacm/HorowitzS74}, the items are divided into two sets $S_1$ and $S_2$
of sizes roughly $\frac12n$.
Initially, the total value and the total weight is computed for every subset of elements of each set $S_i$.
This results in two sets $A,B$, each with $\Oh(2^{n/2})$ pairs of numbers. The algorithm needs to pick a pair from each set
such that the first components of the pairs sum up to at most $W$ and the second components sum up to at least $V$.
This problem can be solved in linear time w.r.t.\ the set sizes provided that the pairs in both sets $A$ and $B$
are sorted by the first component.

Let us introduce a modified version this problem.

\defproblem{$\HS(A,B,w)$}{
  Two sets $A$ and $B$ of points in 2 dimensions and a threshold $w$.
}{
  Do there exist $(x_1,y_1) \in A$, $(x_2,y_2) \in B$ such that $x_1x_2,y_1y_2 \ge w$?
}

A linear-time solution to this problem is the same as for the problem in 
the meet-in-the-middle solution for \Knapsack.
However, for completeness we prove the following lemma
(see also \cite[Lemma 5.6]{DBLP:journals/mst/KociumakaPR19}):

\begin{lemma}[Horowitz and Sahni~\cite{DBLP:journals/jacm/HorowitzS74}]\label{lem:middle}
  The $\HS$ problem can be solved in linear time assuming that the points in $A$ and $B$
  are sorted by the first component.
\end{lemma}

\begin{proof}
  A pair $(x,y)$ is \emph{irrelevant} if there is another pair $(x',y')$ in the same set such that $x' \ge x$ and $y' \ge y$.
  Observe that removing an irrelevant point from $A$ or $B$ leads to an equivalent instance of the $\HS$ problem.

  Since the points in $A$ and $B$ are sorted by the first component,
  a single scan through these pairs suffices to remove all irrelevant elements.
  Next, for each $(x,y)\in A$, the algorithm computes $(x',y')\in B$ such that
  $x'\ge w / x$ and additionally $x'$ is smallest possible.
  As the irrelevant elements have been removed from $B$, this point also maximizes $y'$ among all pairs satisfying
  $x'\ge w / x$.
  If the elements $(x,y)$ are processed by non-decreasing values $x$, the values $x'$ do not increase, and
  thus the points $(x',y')$ can be computed in $\Oh(|A|+|B|)$ time in total.
\end{proof}

\section{Dynamic Programming Algorithm for \WSCS}\label{sec:DP}
Our algorithm is based on dynamic programming.
We start with a less efficient procedure and then improve it in the next section.
Henceforth, we only consider computing the length of the \WSCS; an actual common supersequence of this length can be recovered from the dynamic programming
using a standard approach (storing the parent of each state).

For a weighted string $W$, we introduce a data structure that stores, for every index $i$,
the set $\{\P(S,W[1 \dd i])\,:\,S \in \Match_z(W[1 \dd i])\}$ represented as an array of size at most $z$
(by Observation~\ref{obs:obv}) with entries in the increasing order.
This data structure is further denoted as $\Freq_i(W,z)$.
Moreover, for each element $p \in \Freq_{i+1}(W,z)$ and each letter $c \in \Sigma$,
a pointer to
$p'=p\,/\, \pi^{(W)}_{i+1}(c)$ in $\Freq_{i}(W,z)$ is stored provided that $p' \in \Freq_{i}(W,z)$.
A proof of the next lemma is essentially the same as of Lemma~\ref{lem:mzw}.

\begin{lemma}\label{lemma:Pz}
  For a weighted string $W$ of length $n$, the arrays $\Freq_i(W,z)$, with $i\in [1\dd n]$,
  can be constructed in $\Oh(nz)$ total time if $|\Sigma|=\Oh(1)$.
\end{lemma}
\begin{proof}
 Assume that $\Freq_i(W,z)$ is computed.
 For every $c \in \Sigma$, we create a list
 \[L_c=\{p \cdot \pi^{(W)}_{i+1}(c)\,:\,p \in \Freq_i(W,z),\, p\cdot \pi^{(W)}_{i+1}(c) \ge \fr\}.\]
 The lists are sorted since $\Freq_i(W,z)$ was sorted.
 Then $\Freq_{i+1}(W,z)$ can be computed by merging all the lists $L_c$ (removing duplicates).
 This can be done in $\Oh(z)$ time since $\sigma=\Oh(1)$.
 The desired pointers can be computed within the same time complexity.
\end{proof}

\noindent 
Let us extend the $\WSCS$ problem in the following way:

\defproblem{$\WSCS'(W_1,W_2,\ell,p,q)$:}{
  Weighted strings $W_1,W_2$, an integer $\ell$, and probabilities  $p,q$.
}{
  Is there a string $S$ of length $\ell$ with subsequences $S_1$ and $S_2$ such that $\P(S_1,W_1)=p$ and $\P(S_2,W_2)=q$?
}

In the following, a {\it state} in the dynamic programming denotes a quadruple $(i,j,\ell,p)$, where $i\in [0\dd |W_1|]$, $j\in [0\dd |W_2|]$, $\ell \in [0\dd |W_1|+|W_2|]$, and $p\in \Freq_i(W_1,z)$. 

\begin{observation}\label{obs:n3z}
There are $\Oh(n^3z)$ states.
\end{observation}
In the dynamic programming, for all states $(i,j,\ell,p)$, we compute
\begin{equation}\label{eq:DP}
  \DP[i,j,\ell,p] = \max\{q\,:\,\WSCS'(W_1[1 \dd i],W_2[1 \dd j],\ell,p,q)=\mathbf{true}\}.
\end{equation}



Let us denote $\pi^k_i(c)=\pi^{(W_k)}_i(c)$.
Initially, the array $\DP$ is filled with zeroes, except that the values $\DP[0,0,\ell,1]$ for $\ell\in [0\dd |W_1|+|W_2|]$
are set to 1.
In order to cover corner cases, we assume that $\pi_0^{1}(c)=\pi_0^{2}(c)=1$ for any $c \in \Sigma$ and that $\DP[i,j,\ell,p]=0$ if $(i,j,\ell,p)$ is not a state.
The procedure $\mathsf{Compute}$ implementing the dynamic-programming algorithm is shown as \cref{alg:compute}.

  \begin{algorithm}[ht]
  \setstretch{1.25}
  \caption{$\mathsf{Compute}(W_1,W_2,z)$}\label{alg:compute}
  \For{$\ell=0$ \KwSty{to} $|W_1|+|W_2|$}{
    $\DP[0,0,\ell,1]:=1$\;
  }
  \ForEach{state $(i,j,\ell,p)$ in lexicographic order}{
    \ForEach{$c \in \Sigma$}{
      $x:=\pi^{1}_i(c)$; $y:=\pi^{2}_j(c)$\;
      $\DP[i,j,\ell,p]:=\max\{$\\
      $\quad\DP[i,j,\ell,p],$\\
      $\quad\DP[i-1,j,\ell-1,\frac{p}{x}],$\\
      $\quad y\cdot\DP[i,j-1,\ell-1,p],$\\
      $\quad y\cdot\DP[i-1,j-1,\ell-1,\frac{p}{x}]$\\
      $\}$\;
    }
  }
  \Return{$\min\,\{\ell\;:\; \DP[|W_1|,|W_2|,\ell,p] \ge \fr\ \text{for some}\ p \in \Freq_{|W_1|}(W_1,z)\}$;}
  \end{algorithm}

The correctness of the algorithm is implied by the following lemma:
\begin{lemma}[Correctness of \cref{alg:compute}]\label{lem:corr1}
  The array $\DP$ satisfies \eqref{eq:DP}.
  In particular, we have $\mathsf{Compute}(W_1,W_2,z)=\WSCS(W_1,W_2,z)$.
\end{lemma}
\begin{proof}
  The proof that $\DP$ satisfies \eqref{eq:DP} goes by induction on $i+j$.
  The base case of $i+j=0$ holds trivially.
  It is simple to verify the cases that $i=0$ or $j=0$.
  Let us henceforth assume that $i>0$ and $j>0$.

  We first show that
  \[\DP[i,j,\ell,p] \le \max\{q\,:\,\WSCS'(W_1[1 \dd i],W_2[1 \dd j],\ell,p,q)=\mathbf{true}\}.\]
  The value $q=\DP[i,j,\ell,p]$ was derived from $\DP[i-1,j,\ell-1,p\,/\,x]=q$, or $\DP[i,j-1,\ell-1,p]=q\,/\,y$, or
  $\DP[i-1,j-1,\ell-1,p\,/\,x]=q\,/\,y$,
  where $x = \pi_i^1(c)$ and $y = \pi_j^2(c)$ for some $c\in \Sigma$.
  In the first case, by the inductive hypothesis, there exists a string $T$
  that is a solution to $\WSCS'(W_1[1 \dd i-1],W_2[1 \dd j],\ell-1,p\,/\,x,q)$.
  That is, $T$ has subsequences $T_1$ and $T_2$ such that
  \[\P(T_1,W_1[1 \dd i-1])=p\,/\,x\quad\text{and}\quad\P(T_2,W_2[1 \dd j])=q.\]
  Then, for $S=T c$, $S_1=T_1c$, and $S_2=T_2$, we indeed have 
  \[\P(S_1,W_1[1 \dd i])=p\quad\text{and}\quad\P(S_2,W_2[1 \dd j])=q\]
  The two remaining cases are analogous.

  \medskip 
  Let us now show that
  \[\DP[i,j,\ell,p] \ge \max\{q\,:\,\WSCS'(W_1[1 \dd i],W_2[1 \dd j],\ell,p,q)=\mathbf{true}\}.\]
  Assume a that string $S$ is a solution to $\WSCS'(W_1[1 \dd i],W_2[1 \dd j],\ell,p,q)$.
  Let $S_1$ and $S_2$ be the subsequences of $S$ such that $\P(S_1,W_1)=p$ and $\P(S_2,W_2)=q$.
  
  Let us first consider the case that $S_1[i] = S[\ell] \ne S_2[j]$.
  Then $T_1 = S_1[1\dd i-1]$ and $T_2 = S_2$ are subsequences of $T = S[1\dd \ell-1]$.
  We then have
  \[p' := \P(T_1,W_1[1 \dd i-1]) = p/\pi_{i}^1(S_1[i]).\]
  By the inductive hypothesis, $\DP[i-1,j,\ell-1,p'] \ge q$.
  Hence, $\DP[i,j,\ell,p] \ge q$ because $\DP[i-1,j,\ell-1,p']$ is present as the second argument of the maximum in the dynamic programming algorithm for $c=S[\ell]$.
 
  The cases that $S_1[i] \ne S[\ell] = S_2[j]$ and that $S_1[i] = S[\ell] = S_2[j]$ rely on the values
  $\DP[i,j-1,\ell-1,p] \ge q\,/\,y$ and $\DP[i-1,j-1,\ell-1,p\,/\,x] \ge q\,/\,y$, respectively.

  Finally, the case $S_1[i] \ne S[\ell] \ne S_2[j]$ reduces to one of the previous cases by changing $S[\ell]$ to $S_1[i]$ so that
  $S$ is still a supersequence of $S_1$ and $S_2$ and a solution to $\WSCS'(W_1[1 \dd i],W_2[1 \dd j],\ell,p,q)$.
\end{proof}

\begin{proposition}\label{prop:2}
  The \WSCS problem can be solved in $\Oh(n^3z)$ time if $|\Sigma|=\Oh(1)$.
\end{proposition}
\begin{proof}
  The correctness follows from Lemma~\ref{lem:corr1}.
  As noted in Observation~\ref{obs:n3z}, the dynamic programming has $\Oh(n^3 z)$ states.
  The number of transitions from a single state is constant provided that $|\Sigma| = \Oh(1)$.

  Before running the dynamic programming algorithm of Proposition~\ref{prop:2},
  we construct the data structures $\Freq_i(W_1,z)$ for all $i\in[1\dd n]$ using Lemma~\ref{lemma:Pz}.
  The last dimension in the $\DP[i,j,\ell,p]$ array can then be stored as a position in $\Freq_i(W_1,z)$.
  The pointers in the arrays $\Freq_i$ are used to follow transitions.
\end{proof}

\section{Improvements}\label{sec:improve}
\subsection{First Improvement: Bounds on $\ell$}
Our approach here is to reduce the number of states $(i,j,\ell,p)$ in \cref{alg:compute} from
$\Oh(n^3z)$ to $\Oh(n^2z\log z)$. This is done by limiting the number of values of $\ell$ considered for
each pair of indices $i,j$ from $\Oh(n)$ to $\Oh(\log z)$.

For a weighted string $W$, we define $\H(W)$ as a standard string generated by taking
the most probable letter at each position, breaking ties arbitrarily. The string $\H(W)$ is also called the \emph{heavy} string of $W$.
By $d_H(S,T)$ we denote the Hamming distance of strings $S$ and $T$.
Let us recall an observation from~\cite{DBLP:journals/mst/KociumakaPR19}.
\begin{observation}[{\cite[Observation 4.3]{DBLP:journals/mst/KociumakaPR19}}]\label{obs:3}
  If $S \match{z} W$ for a string $S$ and a weighted string $W$,
  then $d_H(S,\H(W))\le \log_2 z$.
\end{observation}

The lemma below follows from Observation~\ref{obs:3}.
\begin{lemma}\label{lem:gem}
  If strings $S_1$ and $S_2$ satisfy $S_1 \match{z} W_1$ and $S_2 \match{z} W_2$, then
  \[|\SCS(S_1,S_2) - \SCS(\H(W_1),\H(W_2))| \le 2\log_2 z.\]
  %
  %
  %
\end{lemma}
\begin{proof}
  By Observation~\ref{obs:3},
  \[d_H(S_1,\H(W_1)) \le \log_2 z \quad\text{and}\quad d_H(S_2,\H(W_2)) \le \log_2 z.\]
  Due to the relation~\eqref{eq:LCS_SCS} between $\LCS$ and $\SCS$, it suffices to show the following.
  \begin{claim}
    Let $S_1,H_1,S_2,H_2$ be strings such that $|S_1|=|H_1|$ and $|S_2|=|H_2|$.
    If $d_H(S_1,H_1)\le d$ and $d_H(S_2,H_2) \le d$, then
    $|\LCS(S_1,S_2) - \LCS(H_1,H_2)| \le 2d$.
  \end{claim}
  \begin{proof}
Notice that if $S_1',S_2'$ are strings resulting from $S_1,S_2$ by removing up to $d$ letters from each 
of them, then $\LCS(S_1',S_2')\ge \LCS(S_1,S_2)-2d$.

We now create strings $S_k'$ 
for $k=1,2$, by removing  from $S_k$ letters at positions $i$ such that $S_k[i]\ne H_k[i]$. 
Then, according to the observation above, we have 
\[\LCS(S_1',S_2')\ge \LCS(S_1,S_2)-2d.\]
Any common subsequence of $S_1'$ and $S_2'$ is also a common subsequence of 
$H_1$ and $H_2$ since $S_1'$ and $S_2'$ are subsequences of $H_1$ and $H_2$, respectively.
Consequently,
\[\LCS(H_1,H_2)\ge \LCS(S_1,S_2)-2d.\]
In a symmetric way, we can show that $\LCS(S_1,S_2)\ge \LCS(H_1,H_2) - 2d$. This completes the proof
of the claim.
  \end{proof}
\noindent  We apply the claim for $H_1=\H(W_1)$, $H_2=\H(W_2)$, and $d=\log_2 z$.
\end{proof}

Let us make the following simple observation.

\begin{observation}\label{obs:WSCS_SCS}
  If $S=\WSCS(W_1,W_2,z)$, then $S=\SCS(S_1,S_2)$ for some strings $S_1$ and $S_2$ such that $W_1 \R{z} S_1$ and $W_2 \R{z} S_2$.
\end{observation}
Using \cref{lem:gem}, we refine the previous algorithm as shown in \cref{alg:improved1}.

\begin{algorithm}[ht]
  \caption{$\mathsf{Improved1}(W_1,W_2,z)$}\label{alg:improved1}
  \begin{minipage}{\linewidth - \algomargin - \algomargin}
    In the beginning,  we apply the classic $\Oh(n^2)$-time dynamic-programming solution to the standard \SCS
  problem on $H_1 = \H(W_1)$ and $H_2 = \H(W_2)$.
  It computes a 2D array $T$ such that 
  \[T[i,j] = \SCS(H_1[1 \dd i],H_2[1 \dd j]).\]
  
  Let us denote an interval
  \[L[i,j]=[T[i,j] - \lfloor 2\log_2 z \rfloor\ \dd\ T[i,j] + \lfloor 2\log_2 z \rfloor].\]
  We run the dynamic programming algorithm $\mathsf{Compute}$
  restricted to states $(i,j,\ell,p)$ with $\ell \in L[i,j]$.
  
  Let $\DP'$ denote the resulting array, restricted to states satisfying $\ell \in L[i,j]$.
  We return $\min\,\{\ell\;:\; \DP'[|W_1|,|W_2|,\ell,p] \ge \fr\ \text{for some}\ p \in \Freq_{|W_1|}(W_1,z)\}$.

\end{minipage}
\end{algorithm}

\begin{lemma}[Correctness of \cref{alg:improved1}]\label{lem:corr2}
  For every state $(i,j,\ell,p)$, an inequality $\DP'[i,j,\ell,p] \le \DP[i,j,\ell,p]$ holds.
  Moreover, if $S=\SCS(S_1,S_2)$, $|S|=\ell$, $\P(S_1,W_1[1 \dd i])=p \ge \fr$ and $\P(S_2,W_2[1 \dd j])=q \ge \fr$, then $\DP'[i,j,\ell,p] \ge q$.
  Thus, $\mathsf{Improved1}(W_1,W_2,z)=\WSCS(W_1,W_2,z)$.
\end{lemma}

\begin{proof}
  A simple induction on $i+j$ shows that the array $\DP'$ is lower bounded by $\DP$.
  This is because \cref{alg:improved1} is restricted to a subset of states considered by \cref{alg:compute},
  and because $\DP'[i,j,\ell,p]$ is assumed to be 0 while $\DP[i,j,\ell,p]\ge 0$ for states $(i,j,\ell,p)$ ignored in \cref{alg:improved1}.

  We prove the second part of the statement also by induction on $i+j$.
  The base cases satisfying $i=0$ or $j=0$ can be verified easily, so let us
  henceforth assume that $i>0$ and $j>0$. 

  First, consider the case that $S_1[i] = S[\ell] \ne S_2[j]$.
  Let $T=S[1\dd \ell-1]$ and $T_1=S_1[1\dd i-1]$.
  We then have
  \[p' := \P(T_1,W_1[1 \dd i-1]) = p/\pi_{i}^1(S_1[i])\]
  \begin{claim}
    If $S_1[i] = S[\ell] \ne S_2[j]$, then $T=\SCS(T_1,S_2)$.
  \end{claim}
  \begin{proof}
    Let us first show that $T$ is a common supersequence of $T_1$ and $S_2$.
    Indeed, if $T_1$ was not a subsequence of $T$, then $T_1 S_1[i] = S_1$ would not be a subsequence of $T S_1[i] = S$,
    and if $S_2$ was not a subsequence of $T$, then it would not be a subsequence of $T S_1[i] = S$ since $S_1[i] \ne S_2[j]$.

    Finally,  if $T_1$ and $S_2$ had a common supersequence $T'$ shorter than $T$,
    then $T' S_1[i]$ would be a common supersequence of $S_1$ and $S_2$ shorter than $S$.
  \end{proof}

  \noindent
  By the claim and the inductive hypothesis, $\DP'[i-1,j,\ell-1,p'] \ge q$.
  Hence, $\DP'[i,j,\ell,p] \ge q$ due to the presence of the second argument of the maximum in the dynamic programming algorithm for $c=S[\ell]$.
  Note that $(i,j,\ell,p)$ is a state in \cref{alg:improved1} since  $\ell\in L[i,j]$ follows from Lemma~\ref{lem:gem}.

  The cases that $S_1[i] \ne S[\ell] = S_2[j]$ and that $S_1[i] = S[\ell] = S_2[j]$ use the values
  $\DP'[i,j-1,\ell-1,p] \ge q\,/\,y$ and $\DP'[i-1,j-1,\ell-1,p\,/\,x] \ge q\,/\,y$, respectively.
  Finally, the case that $S_1[i] \ne S[\ell] \ne S_2[j]$ is impossible as $S=\SCS(S_1,S_2)$.
\end{proof}

\begin{example}
  Let $W_1=[1,0]$, $W_2=[0]$ (using the notation from Example~\ref{ex:1}), and $z \ge 1$.
  The only strings that match $W_1$ and $W_2$ are $S_1=\mathtt{ab}$ and $S_2=\mathtt{b}$, respectively.
  We have $\DP[2,1,3,1]=1$ which corresponds, in particular, to a solution $S=\mathtt{abb}$ which is not an \SCS of $S_1$ and $S_2$.
  However, $\DP[2,1,2,1]=\DP'[2,1,2,1]=1$ which corresponds to $S=\mathtt{ab}=\SCS(S_1,S_2)$.
\end{example}

\begin{proposition}\label{prop:3}
  The \WSCS problem can be solved in $\Oh(n^2z \log z)$ time if $|\Sigma|=\Oh(1)$.
\end{proposition}
\begin{proof}
  The correctness of the algorithm follows from Lemma~\ref{lem:corr2}.
  The number of states is now $\Oh(n^2 z\log z)$ and thus so is the number of considered transitions.
\end{proof}

\subsection{Second Improvement: Meet in the Middle}
The second improvement is to apply a meet-in-the-middle approach,
which is possible due to following observation resembling Observation~6.6 in~\cite{DBLP:journals/mst/KociumakaPR19}.

\begin{observation}\label{obs:sqrt}
  If $S \match{z} W$ for a string $S$ and weighted string $W$ of length $n$,
  then there exists a position $i \in [1\dd n]$ such that
  \[S[1 \dd i-1] \match{\sqrt{z}} W[1 \dd i-1]\quad\text{and}\quad S[i+1 \dd n] \match{\sqrt{z}} W[i+1 \dd n].\]
\end{observation}
\begin{proof}
  Select $i$ as the maximum index with $S[1 \dd i-1] \match{\sqrt{z}} W[1 \dd i-1]$.
\end{proof}

\noindent
We first use dynamic programming to compute two arrays, $\DPR$ and $\DPL$.
The array $\DPR$ contains a subset of states from $\DP'$; namely the ones that satisfy $p \ge \frac{1}{\sqrt{z}}$.
The array $\DPL$ is an analogous array defined for suffixes of $W_1$ and $W_2$.
Formally, we compute $\DPR$ for the reversals of $W_1$ and $W_2$, denoted as $\DPR^R$, and
set $\DPL[i,j,\ell,p] = \DPR^R[|W_1|+1-i,|W_2|+1-j,\ell,p]$.
Proposition~\ref{prop:3} yields

\begin{observation}
  Arrays $\DPR$ and $\DPL$ can be computed in $\Oh(n^2 \sqrt{z} \log z)$ time.
\end{observation}

Henceforth, we consider only a simpler case in which there exists a solution $S$ to $\WSCS(W_1,W_2,z)$
with a decomposition $S=S_L\cdot S_R$ such that
\begin{equation}\label{eq:simple}
  W_1[1 \dd i] \subseq{\sqrt{z}} S_L \quad\text{and}\quad W_1[i+1 \dd |W_1|] \subseq{\sqrt{z}} S_R
\end{equation}
holds for some $i\in [0\dd |W_1|]$.
 
In the pseudocode, we use the array $L[i,j]$ from the first improvement, denoted here as $\LR[i,j]$, and a symmetric array $\LL$ from right to left, i.e.:
\begin{align*}
  \TT[i,j] &= \SCS(\H(W_1)[i \dd |W_1|],\H(W_2)[j \dd |W_2|]),\\
  \LL[i,j]&=[\TT[i,j] - \floor{2\log_2 z}\dd \TT[i,j] + \floor{2\log_2 z}].
\end{align*}
\cref{alg:improved2} is applied for every $i\in [0\dd |W_1|]$ and $j\in [0\dd |W_2|]$.

\begin{figure}[ht]
  \begin{center}
  \begin{minipage}{0.8\linewidth}
  \begin{algorithm}[H]
  \setstretch{1.3}
  \caption{$\mathsf{Improved2}(W_1,W_2,z,i,j)$}\label{alg:improved2}
  $\mathit{res}:=\infty$\;
  \ForEach{$\ell_L \in \LR[i,j]$, $\ell_R \in \LL[i+1,j+1]$}{
    $A:=\{(p,q)\,:\, \DPR[i,j,\ell_L,p]=q\}$\;
    $B:=\{(p,q)\,:\, \DPL[i+1,j+1,\ell_R,p]=q\}$\;
    \If{$\HS(A,B,z)$}{
      $\mathit{res}:=\min(\mathit{res},\ell_L+\ell_R)$\;
    }
  }
  \Return{$\mathit{res}$};
  \end{algorithm}
  \end{minipage}
\end{center}
\end{figure}

\begin{lemma}[Correctness of \cref{alg:improved2}]\label{lem:corr3}
  Assuming that there is a solution $S$ to $\WSCS(W_1,W_2,z)$ that satisfies \eqref{eq:simple},
 we have \[\WSCS(W_1,W_2,z)=\min_{i,j}(\mathsf{Improved2}(W_1,W_2,z,i,j)).\]
\end{lemma}
\begin{proof}
  Assume that $\WSCS(W_1,W_2,z)$ has a solution $S=S_L\cdot S_R$ that satisfies \eqref{eq:simple}
  for some $i\in [0\dd |W_1|]$ and denote $\ell_L=|S_L|$, $\ell_R=|S_R|$.
  Let $S'_L$ and $S'_R$ be subsequences of $S_L$ and $S_R$ such that
  \[p_L:=\P(S'_L,W_1[1 \dd i]) \ge \tfrac{1}{\sqrt{z}} \quad\text{and}\quad p_R:=\P(S'_R,W_1[i+1 \dd |W_1|]) \ge \tfrac{1}{\sqrt{z}}.\]
  Let $S''_L$ and $S''_R$ be subsequences of $S_L$ and $S_R$ such that
  \[\P(S''_L,W_2[1 \dd j])=q_L \quad\text{and}\quad \P(S''_R,W_2[j+1 \dd |W_2|])=q_R\]
  for some $j$ and $q_Lq_R \ge \fr$.

\smallskip
  By Lemma~\ref{lem:corr2}, $\DPR[i,j,\ell_L,p_L] \ge q_L$ and $\DPL[i+1,j+1,\ell_R,p_R] \ge q_R$.
  Hence, the set $A$ will contain a pair $(p_L,q'_L)$ such that $q'_L \ge q_L$ and the set $B$ will contain a pair $(p_R,q'_R)$ such that $q'_R \ge q_R$.
  Consequently, $\HS(A,B,z)$ will return a positive answer.

  Similarly, if $\HS(A,B,z)$ returns a positive answer for given $i$, $j$, $\ell_L$ and $\ell_R$, then
  \[\DPR[i,j,\ell_L,p_L] \ge q_L\quad\text{and}\quad\DPL[i+1,j+1,\ell_R,p_R] \ge q_R\]
 for some $p_Lp_R,q_Lq_R \ge \fr$.
  By Lemma~\ref{lem:corr2}, this implies that
  \[\WSCS'(W_1[1 \dd i],W_2[1 \dd j],\ell_L,p_L,q_L)\]
  and
  \[\WSCS'(W_1[i+1 \dd |W_1|],W_2[j+1 \dd |W_2|],\ell_R,p_R,q_R)\]
  have a positive answer, so
  \[\WSCS'(W_1,W_2,\ell_L+\ell_R,p_Lp_R,q_Lq_R)\]
  has a positive answer too.
  Due to $p_Lp_R,q_Lq_R \ge \fr$, this completes the proof.
%
\end{proof}

\begin{proposition}\label{prop:4}
  The \WSCS problem can be solved in $\Oh(n^2\sqrt{z} \log^2 z)$ time if $|\Sigma|=\Oh(1)$.
\end{proposition}
\begin{proof}
  We use the algorithm $\mathsf{Improved2}$, whose correctness follows from Lemma~\ref{lem:corr3} in case \eqref{eq:simple} is satisfied. 
The general case of Observation~\ref{obs:sqrt} requires only a minor technical change to the algorithm.
Namely, the computation of $\DPR$ then additionally includes all states $(i,j,\ell,p)$ such that $\ell\in \LR[i,j]$,
$p\ge \fr$, and $p=\pi^1_i(c)p'$ for some $c\in \Sigma$ and $p'\in \Freq_{i-1}(W_1,\sqrt{z})$.
Due to $|\Sigma| = \Oh(1)$, the number of such states is still $\Oh(n^2 \sqrt{z}\log z)$.

For every $i$ and $j$, the algorithm solves $\Oh(\log^2 z)$ instances of \HS, each of size $\Oh(\sqrt{z})$.
This results in the total running time of $\Oh(n^2\sqrt{z} \log^2 z)$.
\end{proof}

\subsection{Third Improvement: Removing one $\log z$ Factor}
The final improvement is obtained by a structural transformation after which we only need to consider $\Oh(\log z)$ pairs $(\ell_L,\ell_R)$.

For this to be possible, we compute prefix maxima on the $\ell$-dimension of the $\DPR$ and $\DPL$ arrays in order to guarantee monotonicity.
That is, if $\HS(A,B,z)$ returns true for $\ell_L$ and $\ell_R$, then we make sure that it would also return true if any of these two lengths increased (within the corresponding intervals).

\begin{algorithm}[t!]
  \setstretch{1.25}
  \SetKwComment{Comment}{$\triangleright$\ }{}
  \caption{$\mathsf{Improved3}(W_1,W_2,z,i,j)$}\label{alg:improved3}
  \ForEach{state $(i,j,\ell,p)$ of $\DPR$ in lexicographic order}{
    $\DPR[i,j,\ell,p]:=\max(\DPR[i,j,\ell,p],\DPR[i,j,\ell-1,p])$\;
  }
  \ForEach{state $(i,j,\ell,p)$ of $\DPL$ in lexicographic order}{
    $\DPL[i,j,\ell,p]:=\max(\DPL[i,j,\ell,p],\DPL[i,j,\ell-1,p])$\;
  }
  $[a\dd b] := \LR[i,j]$; $[a'\dd b']:=\LL[i+1,j+1]$\;
  $\ell_L := a$; $\ell_R := b' + 1$; $\mathit{res}:=\infty$\;
  \While{$\ell_L \le b$ \KwSty{and} $\ell_R \ge a'$}{
    $A:=\{(p,q)\,:\, \DPR[i,j,\ell_L,p]=q\}$\;
    $B:=\{(p,q)\,:\, \DPL[i+1,j+1,\ell_R - 1,p]=q\}$\;
    \If(\Comment*[f]{$\ell_R$ is too large for the current $\ell_L$}){$\HS(A,B,z)$}{
      $\ell_R:=\ell_R - 1$\;
    }\Else(\Comment*[f]{$\ell_R$ reached the target value for the current $\ell_L  $}){
      \lIf{$\ell_R \le  b'$}{
        $\mathit{res}:=\min(\mathit{res},\ell_L+\ell_R)$}
      $\ell_L:=\ell_L+1$\;
    }
  }
  \Return{$\mathit{res}$};
  \end{algorithm}

  This lets us compute, for every $\ell_L\in \LR[i,j]$ the smallest $\ell_R\in \LL[i,j]$ such that
  $\HS(A,B,z)$ returns true using $\Oh(\log z)$ iterations because the sought $\ell_R$ may only decrease
  as $\ell_L$ increases. 
  The pseudocode is given in \cref{alg:improved3}.

\begin{theorem}\label{thm:main}
  The \WSCS problem can be solved in $\Oh(n^2\sqrt{z} \log z)$ time
  if $|\Sigma|=\Oh(1)$.
\end{theorem}
\begin{proof}
  Let us fix indices $i$ and $j$.
  Let us denote $\Freq_i(W,z)$ by $\RFreq_i(W,z)$ and introduce a symmetric array
  \[\LFreq_i(W,z)=\{\P(S,W[i \dd |W|])\,:\,S \in \Match_z(W[i \dd |W|])\}.\]
  In the first loop of prefix maxima computation, we consider all $\ell \in \LR[i,j]$ and $p \in \RFreq_i(W_1,\sqrt{z})$,
  and in the second loop, all $\ell \in \LL[i,j]$ and $p \in \LFreq_{i}(W_1,\sqrt{z})$.
  Hence, prefix maxima take $\Oh(\sqrt{z}\log{z})$ time to compute.

  Each step of the while-loop in $\mathsf{Improved3}$ increases $\ell_L$ or decreases $\ell_R$.
  Hence, the algorithm produces only $\Oh(\log z)$ instances of \HS, each of size $\Oh(\sqrt{z})$.
  The time complexity follows.
\end{proof}

\section{Lower Bound for \WLCS}
Let us first define the \WLCS problem 
as it was stated in \cite{DBLP:journals/jda/AmirGS10,DBLP:journals/dam/CyganKRRW16}.

\defproblem{\WLCSFull ($\WLCS(W_1,W_2,z)$)}{
  Weighted strings $W_1$ and $W_2$ of length up to $n$ and a threshold $\fr$.
}{
  A longest standard string $S$ such that $S\subseq{z} W_1$ and $S\subseq{z}W_2$.
}

\medskip
We consider the following well-known NP-complete problem~\cite{DBLP:conf/coco/Karp72}:

\defproblem{ Subset Sum }{
  A set $S$ of positive integers and a positive integer $t$. 
}{
  Is there a subset of $S$ whose elements sum up to $t$?
}
\begin{theorem}\label{WLCS hardness}
 The \WLCS problem cannot be solved in $\Oh(n^{f(z)})$ time if $\mathrm{P}\ne\mathrm{NP}$.
\end{theorem}
\begin{proof}
We show the hardness result by reducing the NP-complete \textsc{Subset Sum} problem to the \WLCS problem with a constant value of $z$.

For a set $S=\{s_1,s_2,\ldots,s_n\}$ of $n$ positive integers, a positive integer $t$, and an additional parameter $p\in [2\dd n]$,
we construct two weighted strings $W_1$ and $W_2$ over the alphabet $\Sigma=\{\mathtt{a},\mathtt{b}\}$, each of length $n^2$.

Let $q_i=\frac{s_i}{t}$.  
At positions $i\cdot n$, for all $i=[1\dd n]$, the weighted string $W_1$ contains letter $\mathtt{a}$ with probability $2^{-q_i}$
and $\mathtt{b}$ otherwise, while $W_2$ contains $\mathtt{a}$ with probability $2^{\frac{1}{p-1}(q_i-1)}$ and $\mathtt{b}$ otherwise. 
All the other positions contain letter $\mathtt{b}$ with probability $1$. We set $z = 2$.
 
We assume that $S$ contains only elements smaller than $t$ (we can ignore the larger ones and if there is an element equal to $t$, then there is no need for a reduction). 
All the weights of $\mathtt{a}$ are then in the interval $(\frac{1}{2},1)$ since $-q_i\in (-1,0)$ and $\frac{1}{p-1}(q_i-1) \in (-1,0)$.
Thus, since $z=2$, letter $\mathtt{b}$ originating from a position $i\cdot n$ can never occur in a subsequence of $W_1$ or in a subsequence of $W_2$. 
Hence, every common subsequence of $W_1$ and $W_2$ is a subsequence of $(\mathtt{b}^{n-1}\mathtt{a})^n$.
 
For $I\subseteq[1\dd n]$, we have
\begin{align*}
  \prod_{i\in I}\pi^{(W_1)}_{i\cdot n}(\mathtt{a}) = \prod_{i \in I} 2^{-s_i/t} \ge2^{-1} =\tfrac{1}{z}\ \Longleftrightarrow&\ \sum_{i\in I}s_i\le t\\ 
  \text{and}\hspace{8cm}\\
\prod_{i\in I}\pi^{(W_2)}_{i\cdot n}(\mathtt{a}) = \prod_{i \in I}2^{\frac{1}{p-1}(s_i/t-1)}\,\ge \,2^{-1}=\tfrac{1}{z}\ \Longleftrightarrow&\\
\tfrac{1}{t(p-1)}\left(\sum_{i\in I}s_i\right)-\tfrac{|I|}{p-1}\,\ge-1\ \Longleftrightarrow&\ \sum_{i\in I}s_i\ge t(1-p+|I|).
\end{align*}
 
 If $I$ is a solution to the instance of the \textsc{Subset Sum} problem, then for $p=|I|$ there is a weighted common subsequence of length $n(n-1)+p$ 
 obtained by choosing all the letters $\mathtt{b}$ and the letters $\mathtt{a}$ that correspond to the elements of $I$.

 Conversely, suppose that the constructed \WLCS instance with a parameter $p\in [2\dd n]$ has a solution of length at least $n(n-1)+p$.
 Notice that $\mathtt{a}$ at position $i\cdot n$ in $W_1$
 may be matched against $\mathtt{a}$ at position $i'\cdot n$ in $W_2$ only if $i=i'$.
 (Otherwise, the length of the subsequence would be at most $(n-|i-i'|)n\le (n-1)n<n(n-1)+p$.)
 Consequently, the solution yields a subset $I\subseteq [1\dd n]$ of at least $p$ indices
 $i$ such that $\mathtt{a}$ at position $i\cdot n$ in $W_1$ is matched against 
 $\mathtt{a}$ at position $i\cdot n$ in $W_2$.
 By the relations above, we have (a) $|I| \geq p$, (b) $\sum_{i\in I}s_i\le t$, and (c) $\sum_{i\in I}s_i\ge t(1-p+|I|)$.
 Combining these three inequalities, we obtain $\sum_{i\in I}s_i= t$ and conclude that 
 the \textsc{Subset Sum} instance has a solution.
 
 Hence, the \textsc{Subset Sum} instance has a solution if and only if there exists $p\in[2\dd n]$ such that the constructed \WLCS  instance with $p$ has a solution of length at least $n(n-1)+p$.
 This concludes that an $\Oh(n^{f(z)})$-time algorithm for the \WLCS problem implies the existence of an $\Oh(n^{2f(2)+1})=\Oh(n^{\Oh(1)})$-time algorithm for the \textsc{Subset Sum} problem. The latter would yield $P=NP$.
\end{proof}

\begin{example}
 For $S=\{3,7,11,15,21\}$ and $t=25=3+7+15$, both weighted strings $W_1$ and $W_2$ 
 are of the form:
\[\mathtt{b^4\,*\,b^4\,*\,b^4\,*\,b^4\,*\,b^4\,*}\,,\]
 where each $\mathtt{*}$ is equal to either $\mathtt{a}$ or $\mathtt{b}$ with different probabilities.

The probabilities of choosing $\mathtt{a}$'s for $W_1$ are equal respectively to 
\[\big(2^{-\frac{3}{25}},2^{-\frac{7}{25}},2^{-\frac{11}{25}},2^{-\frac{15}{25}},2^{-\frac{21}{25}}\big),\]
 while for $W_2$ they depend on the value of $p$, and are equal respectively to
 \[\big(2^{-\frac{22}{25(p-1)}},2^{-\frac{18}{25(p-1)}},2^{-\frac{14}{25(p-1)}},2^{-\frac{10}{25(p-1)}},2^{-\frac{4}{25(p-1)}}\big).\]
For $p=3$, we have: $\WLCS(W_1,W_2,2)\,=\,\mathtt{b^4\,a\,b^4\,a\,b^4\,b^4\,a\,b^4}$, which corresponds to taking the first, the second, and the fourth $\mathtt{a}$. 
The length of this string is equal to $23=n(n-1)+p$, 
and its probability of matching is $\frac12 =  2^{-\frac{22}{50}} \cdot 2^{-\frac{18}{50}} \cdot 2^{-\frac{10}{50}}$.
Thus, the subset $\{3,7,15\}$ of $S$ consisting of its first, second, and fourth element 
is a solution to the \textsc{Subset Sum} problem.
\end{example}

\bibliographystyle{plainurl}
\bibliography{wscs}

\end{document}